%% file: arxiv.tex
\documentclass[a4paper,thm-restate,cleveref, autoref]{lipics-v2021}
\nolinenumbers 
\usepackage{caption}
\usepackage{subcaption}
\usepackage{amssymb}
\usepackage{url}
\usepackage{hyperref}
\usepackage{float}
\usepackage{algorithm}
\usepackage{graphicx}

\nolinenumbers 
\usepackage{caption}
\usepackage{subcaption}
\usepackage{amssymb}
\usepackage{url}
\usepackage{hyperref}
\usepackage{float}
\usepackage{algorithm}
\usepackage{graphicx}

\hideLIPIcs

\newcommand{\braket}[1]{\left< #1 \right>}
\def\D{\mathrm{d}}
\newcommand{\RR}{\mathbb{R}}

\newcommand{\NN}{\mathbb{N}}

\newcommand{\curve}{\gamma}

\newcommand{\rfreespace}{\mathcal{R}}
\newcommand{\freespace}{\mathcal{D}}
\newcommand{\freespacediagram}{\mathrm{FSD}}
\newcommand{\distf}{\mathrm{d}_{\mathcal{F}}}
\newcommand{\regs}{\mathcal{R}}
\newcommand{\sings}{I_{\text{sing}}}
\newcommand{\intersections}{\mathcal{I}}
\newcommand{\subcell}{\mathcal{S}}
\DeclareMathOperator{\simp}{simp}

\usepackage[dvipsnames,table,xcdraw]{xcolor}

\def\eps{\epsilon}
\def\vareps{\varepsilon}

\usepackage{ulem}

\usepackage[utf8]{inputenc}

\usepackage{amsmath,amsfonts,amssymb,xspace}
\usepackage{ mathrsfs }
\usepackage{amsthm}
\usepackage{graphicx}
\usepackage{paralist}
\usepackage{mathtools}
\usepackage{xcolor}

\usepackage{algorithm}
\usepackage[noend]{algpseudocode}
\usepackage{xifthen}

\usepackage{comment}

\title{Revisiting the Fréchet distance between piecewise smooth curves}
\author{Jacobus Conradi}{University of Bonn, Bonn, Germany}{jacobus.conradi@gmx.de}{https://orcid.org/0000-0002-8259-1187}{Partially funded by the Deutsche Forschungsgemeinschaft (DFG, German Research Foundation) - 313421352 (FOR 2535 Anticipating Human Behavior). Affiliated with Lamarr Institute for Machine Learning and Artificial Intelligence.}
\author{Anne Driemel}{Hausdorff Center for Mathematics, University of Bonn, Germany}{driemel@cs.uni-bonn.de}{https://orcid.org/0000-0002-1943-2589}{Affiliated with Lamarr Institute for Machine Learning and Artificial Intelligence.}
\author{Benedikt Kolbe}{Hausdorff Center for Mathematics, University of Bonn, Bonn, Germany}{benedikt.kolbe@physik.hu-berlin.de}{}{Partially funded by the Deutsche Forschungsgemeinschaft (DFG, German Research Foundation) – 459420781.}

\authorrunning{Conradi, Driemel, and Kolbe}

\Copyright{Jacobus Conradi, Anne Driemel, Benedikt Kolbe} 



\begin{document}

\maketitle

\begin{abstract}
\input{abstract}
\end{abstract}


\section{Introduction and motivation}\label{sec:intro}
The Fréchet distance is a well-studied distance measure between curves, with a long history in both applications and algorithmic research.
The wealth of work surrounding the analysis of algorithms for computing the Fréchet distance is centered primarily on polygonal curves. However, more complicated curves and especially splines have become commonplace in industrial applications for, e.g., computer graphics, robotics and to represent motion tracking or planning data. In such applications, the number of dimensions corresponds to the number of tracked parameters, leading to a high-dimensional ambient space. On the other hand, polygonal curves in applications often represent (unnecessarily complex) discretizations approximating an underlying smooth curve that can be described more efficiently as a smooth curve using only a small number of parameters. For an algorithmic example, smooth curves may help in more naturally describing the average of a set of polygonal curves. A crucial prerequisite to using smooth curves similarly to polygonal curves in such contexts is the ability to effectively answer elementary algorithmic questions for such curves. A natural and fundamental task in computational geometry is therefore to devise algorithms for the computation of the Fréchet distance between smooth curves such as splines. Despite this, as far as we know, there is no known approach to realizing such a computation for curves in $\RR^d$. To tackle the case of smooth curves in the plane ($d=2$),  Rote~\cite{Rote2007} introduced an approach based on analyzing the turning angle and planar curvature of the planar curves. However, this approach does not easily generalize to higher dimensions. We revisit this problem and present a novel, simpler approach, with the additional benefit that it works for higher dimensions, with the same time complexity. Our methods are conceptually simple, but rely on a number of key technical ingredients. 

\subparagraph{{Problem definition}}
Throughout the paper, $\curve_1$ and $\curve_2$ will be used to denote two piecewise smooth curves in $\RR^d$ with $d$ fixed, that is, continuous maps $\curve_1,\curve_2:[0,1]\to \RR^d$ that are comprised of $m$ and $n$ smooth pieces, each of class $C^2$. 
Let $A_{[0,1]}$ be the set of continuous and bijective maps $\alpha:[0,1]\to [0,1]$ that are increasing.
The \textbf{Fréchet distance} between $\curve_1$ and $\curve_2$ is defined as $\distf(\curve_1,\curve_2):=\inf\limits_{\alpha,\beta\in A_{[0,1]}} \max\limits_{t\in [0,1]} \| \curve_1(\alpha(t))-\curve_2(\beta(t))\|$.
Our methods naturally allow any fixed $\ell_p$ norm with $1<p<\infty$ for the norm $\|\cdot \|$ (the cases $p=1,\infty$, while possible, would add a level of technicality to our treatment that distracts from its relative simplicity). 
To compute $\distf(\curve_1,\curve_2)$, we mostly focus on the \textbf{decision problem} of deciding whether the Fréchet distance between two piecewise smooth curves is at most a given $\delta>0$.

\subparagraph{Results}
Our first main contribution is that we establish an algorithm to solve the decision problem for the Fréchet distance between piecewise smooth curves. Assuming that the curves are \textbf{algebraically bounded curves}, i.e., piecewise smooth algebraic curves where the degree of the curves is bounded by a constant, we obtain a bound of $O(mn)$ for the time complexity of the decision problem, which matches the polygonal case. The running time is independent of the ambient dimension but the algebraic complexity of the operations involved in the algorithm depends on the dimension and the nature of the curves. 
Our algorithm for the decision problem results in an algorithm for the computation of the Fréchet distance for algebraically bounded curves in $O(mn\log(mn))$ time using parametric search, similarly to the polygonal case. 

It is known~\cite{Bringmann2014} that the decision problem cannot be solved in strongly subquadratic time, so research has focused on investigating algorithms for restricted classes of curves. Our second contribution is that we show that we can adapt the framework from~\cite{Driemel2012} for an efficient $(1+\eps)$-approximation algorithm for the Fréchet distance between two $c$-packed, polyognal curves to the setting of $c$-packed piecewise smooth curves in $\RR^d$. To this end, we introduce a simplification procedure for piecewise smooth curves and distill the necessary ingredients to obtain a linear time decision algorithm for algebraically bounded $c$-packed curves.

\subparagraph{Comparison to previous work}
To arrive at an algorithm for the decision problem for smooth planar curves for the $\ell_2$-norm for a given $\delta$ in general position, Rote uses a partitioning of the smooth curves to obtain pieces for which the associated \textbf{free space diagram} $\freespacediagram_{\delta}$ (Section~\ref{sec:combdescfreespace}) is well-behaved. His approach relies on analyzing the curves $\curve_1$ and $\curve_2$, introducing cuts to control the turning angle and at points with a certain value for the planar curvature. In contrast to this, our approach is to analyze the free space diagram directly, by studying the boundary of the \textbf{free space} $\freespace_{\delta}$ in $\freespacediagram_{\delta}$. The free space is defined as the set of parameter value pairs at which the curves are at most a distance of $\delta$ apart. Working directly with the free space not only leads to a conceptually simpler algorithm, but also avoids relatively complicated integral equations related to the curvature and the turning angle. We propose a refined decomposition of each cell of $\freespacediagram_{\delta}$ into a controlled number (depending on the degree of the curves) of subcells (Section~\ref{sec:combdescfreespace}), for which the existence of a monotone path connecting two intervals on the boundary of a subcell can be read off. Here, the role of convexity of the free space in a cell for polygonal curves is replaced by monotonicity of the boundary curves of $\freespace_{\delta}$ within each subcell of the refined decomposition. 
We emphasize here that our construction of the refined decomposition exclusively accesses the same values that are also required in Rote's work to process each subcell of $\freespacediagram_{\delta}$.

Unlike the polygonal case, the free space within a cell of $\freespacediagram_{\delta}$ can be very complicated, as illustrated by a contour plot of the distance function in parameter space for two degree 3 splines in $\RR^3$ in Figure~\ref{fig:contourplot} for different values of $\delta$.
\begin{figure}[t]
  \centering
   \includegraphics[width=\textwidth]{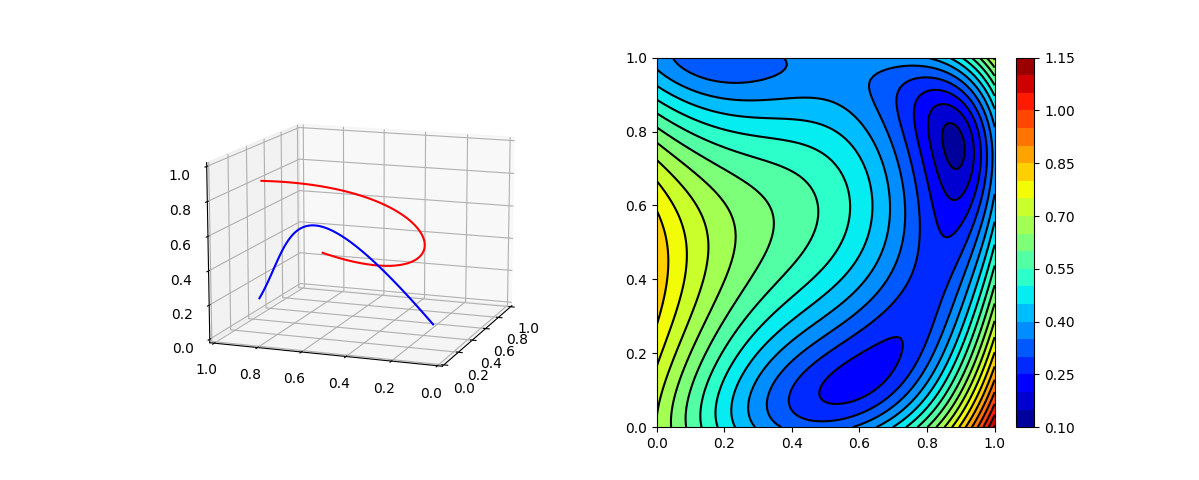}
      \caption{Two smooth curves in $\RR^3$ and a contour plot of the associated distance function in the joint parametric space of the curves. }\label{fig:contourplot}
\end{figure}
Figure~\ref{fig:freespacedecomp} shows another example of the kind of behavior of the free space one can expect within a cell. 

We note that both Rote's decision algorithm as well as ours assume values of $\delta$ for which the boundary of $\freespace_{\delta}$ has no singularities. A key technical result we show (Section~\ref{sec:singularitystructure}) is that singularities of the boundary of $\freespace_{\delta}$ are confined to a small number of critical values of $\delta$ and are not necessary for the computation of $\distf$. Together with the approximation schemes for smooth curves we introduce later on, our results show that the setting of smooth curves provides a natural extended framework for algorithmic approaches to the Fréchet distance.

\subparagraph{Computational model}
We make the same assumptions as Rote in his work on planar smooth curves, subsuming the real RAM model. In essence, we assume that we can compare two real solutions of an algebraic equation. Specifically, we require being able to compute the intersection of a curve with a sphere of a given radius centered at a point of another curve and find the parameter values in $[0,1]$ that correspond to the intersections. While we limit the use of such algebraic operations where possible, our approach to the decision problem via construction of a specific partitioning of the free space diagram, the subsequent computation of $\distf$, as well as the simplification ultimately make use of this assumption. Studying the algebraic numbers involved in the computations for different classes of curves goes beyond the scope of this paper. 

\section{A combinatorial description of the free space diagram}\label{sec:combdescfreespace}
For two piecewise smooth curves $\curve_1,\curve_2 :[0,1]\to \RR^d$ consisting of $m$ and $n$ pieces, respectively, and $\delta>0$, the \textbf{free space} $\freespace_{\delta}=\freespace_{\delta}(\curve_1,\curve_2)$ is defined as 
\[
\freespace_{\delta}(\curve_1,\curve_2)=\left\{(x,y)\in [0,1]^2 |\|\curve_1(x)-\curve_2(y)\|\le \delta \right\}.
\]
The complement of $\freespace_{\delta}$ in $[0,1]^2$ is commonly referred to as the \textbf{forbidden region}.

There is a natural partition of the joint parameter space $[0,1]^2$ of both curves into $m\cdot n$ rectangular cells such that $\curve_1$ and $\curve_2$ are smooth when restricting to the interior of each rectangle. We refer to the resulting decomposition of $[0,1]^2$ together with the partitioning into the free space and forbidden region as the \textbf{free space diagram} $\freespacediagram_{\delta}$. The primary motivation behind the introduction of the free space diagram is the elementary observation that $\distf(\curve_1,\curve_2)\le \delta$ iff there is a path from $(0,0)$ to $(1,1)$ through the free space in $[0,1]^2$ that is monotone in both coordinates.

\subsection{Overview of the algorithm}

Similarly to the classical polygonal case, to solve the decision problem, we investigate the existence of a monotone (in both coordinates) path from $(0,0)$ to $(1,1)$ in the free space $\freespace_{\delta}$. To this end, we refine the free space diagram using the boundary $B_{\delta}$ of the free space. 
 Our decision algorithm has the following high-level description. 
\begin{enumerate}
    \item Mark the minima and maxima of the boundary $B_{\delta}$ of the free space in $\freespacediagram_{\delta}$ in the $x$ (horizontal) and $y$ (vertical) direction.
    \item Cut each cell of $\freespacediagram_{\delta}$ into subcells, horizontally (vertically) through each marked point if it has a vertical (horizontal) tangent. Mark each point of intersection of a cut with $B_{\delta}$.
    \item For each resulting subcell, pair the marked points on the boundary according to how they are connected by $B_{\delta}$ through monotone arcs, so that adjacent points are paired. 
    \item Solve the decision problem for $\freespacediagram_{\delta}$ using only the marked points and pairings by computing reachable intervals on the boundaries of cells, in particular
    \begin{enumerate}\vspace{-0.3\baselineskip}
        \item process all cells in lexicographical order of their indices (row by row, from the left);
        \item for each cell, process all subcells within the cell in lexicographical order.
    \end{enumerate}
\end{enumerate}

In our analysis, we make certain assumptions on the given $\delta$ with respect to the input curves, which can be ensured by applying a random small perturbation. We follow up our presentation of the decision algorithm with a justification for why this assumption does not lead to loss of generality in the context of using the decision algorithm in practice and to compute the Fréchet distance (Section~\ref{sec:singularitystructure}).

 \subsection{Refining the free space diagram}\label{sec:freespacepartition}

We consider the boundary $B_{\delta}$ of $\freespacediagram_{\delta}$ as a set of curves, as opposed to the boundary of a region. 
We define the three sets $E_h,E_v$ and $\sings$ of points on $B_{\delta}$: 
\begin{enumerate}
\item The set $E_h\subset B_{\delta}$ of extrema of the free space in the $y$ direction (with horizontal tangent). 
\item The set $E_v\subset B_{\delta}$ of extrema of the free space in the $x$ direction (with vertical tangent).
\item The set $\sings$ of singularities of $B_{\delta}$ in the interior of the cells of $\freespacediagram_{\delta}$, consisting of points where $B_{\delta}$ is not $C^1$, which includes points of self-intersection and cusps.
\end{enumerate}
The geometry of $B_{\delta}$ is related to the geometry of equidistant curves and can be quite complicated in general. In Section~\ref{sec:singularitystructure}, we show that for almost all $\delta$, there are no singular points of $B_{\delta}$ in the interior of each cell in $\freespacediagram_{\delta}$ associated to the smooth pieces of the curves, so that $\sings=\emptyset$, after possibly applying a small perturbation to $\delta$. We note that for the norms $\ell_1$ and $\ell_{\infty}$ in the definition of the Fréchet distance, $B_{\delta}$ may contain cusp singularities for all values of $\delta$ in an open interval.

Using the points in $E_h$, $E_v$  (and $\sings$), we define a partition of each cell of the free space diagram $\freespacediagram_{\delta}$ that lends itself to combinatorial investigations into the decision problem. 
For simplicity of exposition, we assume that each of $E_h$, $E_v$ and $\sings$ is a collection of isolated points. In particular, $B_{\delta}$ does not have a vertical or horizontal segment, which means that there is no arc of one curve that lies at a constant distance $\delta$ from a point on the other. If there are such circular arcs in any of the three sets, then in the following we represent each such arc by a single point lying on it, distinct from others in the sets.  

For a point $z\in E_h$ ($E_v$), we fix the cell in $\freespacediagram_{\delta}$ containing $z$, and trace the vertical (horizontal) line incident to $z$ inside this cell. For every point $z_s\in \sings$, similarly trace a horizontal line incident to $z_s$. 
The result is a refinement of each cell of $\freespacediagram_\delta$ into a collection of \textbf{subcells} $\{S\}$. Figure~\ref{fig:freespacedecomp} shows an illustration of the refinement inside a cell of $\freespacediagram_\delta$.
\begin{lemma}\label{lem:montonearcs}
The boundary $B_{\delta}$ of the free space in the interior of each subcell in $\{S\}$ is a union of smooth arcs that are monotone in both coordinates of $\RR^2$ and are disjoint except possibly at the boundary of a subcell.
\end{lemma}
\begin{proof}[Proof of Lemma~\ref{lem:montonearcs}]\label{proof:lem:monotonearcs}
Consider a curve $\beta$ representing a connected component of $B_{\delta}$ inside the interior of a subcell $\subcell$. Observe that $\beta$ cannot intersect another arc of $B_{\delta}$ in the interior of $\subcell$ and that $\beta$ is at least $C^1$, as the singular points of $B_{\delta}$ are confined to the boundaries of the subcells in $\{S\}$. 

If $\beta$ is not monotone in either coordinate, consider the case that $\beta$ starts on the bottom edge of $\subcell$ and is not vertical and thus locally given as a function of $x$. If $\beta$ is not monotone in the $x$-coordinate, then there are $x_1\neq x_2$ such that $\frac{d\beta(x_1)}{dx}=\beta'(x_1)>0>\beta'(x_2)$. By continuity, there is a point $x^*\in (x_1,x_2)$ with $\beta'(x^*)=0$, whence $x^*\in E_h$. By construction, $x^*$ then lies on the vertical boundary of $\subcell$. Other cases for $\beta$ can be treated similarly. 
\end{proof}

We record each intersection $\intersections_{\subcell}$ of $B_{\delta}$ with the boundary of each subcell $\subcell$, which together form the set $\intersections=\bigcup_{\subcell \text{ is subcell}}\intersections_{\subcell}$ of all intersections of subcell walls with $B_{\delta}$.

The construction of $\intersections$ consists of finding the intersections of a horizontal or vertical line in $\freespacediagram_{\delta}$ with $B_{\delta}$, which amounts to computing the intersections of a given smooth piece of $\curve_1$, with a sphere of radius $\delta$ centered at a point of $\curve_2$ (and vice versa, switching the roles of $\curve_1$ and $\curve_2$). The associated values for $(x,y)$ in $\freespacediagram_{\delta}$ correspond to the parameter values for the computed points on each curve. 

In the following, we impose the general position assumption that $\delta$ is such that $\sings=\emptyset$ and postpone a justification for this assumption to Section~\ref{sec:singularitystructure}.
The boundary $B_{\delta}$ can be naturally interpreted as a graph $G_{\delta}$ with vertex set $\intersections$. Each edge of $G_{\delta}$ is a monotone arc contained in a subcell, by Lemma~\ref{lem:montonearcs}. In the remainder of this section we focus on elucidating two crucial observations regarding $G_{\delta}$. The first is that the \textbf{combinatorial structure} of $G_{\delta}$, i.e., which points of $\intersections_{\subcell}$ are connected by monotone edges in $B_{\delta}$, can be deduced from the set $\intersections_{\subcell}$ along with what we call slope information at certain points, which is readily computable additional information which we introduce more carefully below. The second observation is that the answer to the decision problem does not depend on the monotone arcs joining points on subcell walls. In other words, the combinatorial structure of $G_{\delta}$ is sufficient to answer the decision problem, which justifies referring to $G_{\delta}$ as a combinatorial model for $\freespacediagram_{\delta}$. We elaborate on the necessary adjustments to the existing framework for the decision algorithm for polygonal curves in more detail in Subsection~\ref{sec:monotonealgo}.

To begin with, we partition the two sets $E_h$ and $E_v$ into the sets $E_h^+$ and $E_h^-$, and $E_v^+$ and $E_v^-$, respectively, according to whether the forbidden region lies locally to the right of or above the point ($-$), or to the left of or below the point ($+$). 
For the bottom and left edge of each subcell $\subcell$, we refer to the information of whether the boundary $B_{\delta}$ at each point in $\intersections_{\subcell}$ is increasing or decreasing as a function of the horizontal $x$-coordinate as the \textbf{slope information} of these points. In other words, the slope information at a point $z\in \intersections_{\subcell}$ can be thought of an extra bit associated to $z$ that 
encodes whether $B_{\delta}$ curves to the left or to the right at $z$. 
In Figure~\ref{fig:freespacedecomp}, the slope information is illustrated by an arrow pointing to the left or right (top or bottom) next to the point on the bottom (left) edge of $\freespace_{\delta}$. 

\begin{figure}[t]
  \centering
   \includegraphics[width=\textwidth]{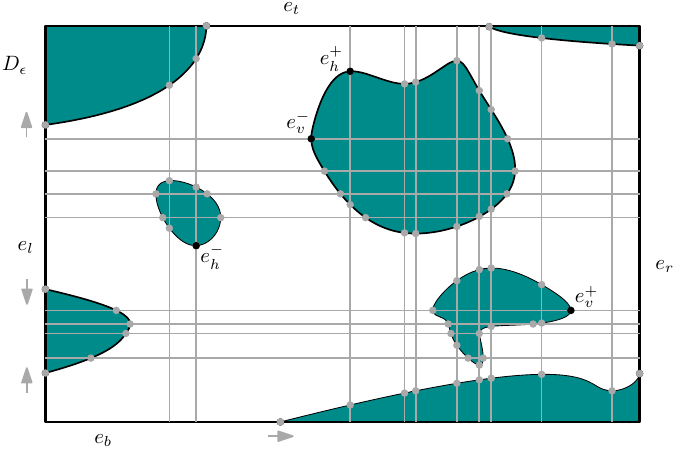}
      \caption{The decomposition of a cell of the free space diagram into subcells arising from the horizontal and vertical lines at extremities of the cyan forbidden region in the coordinate directions.}\label{fig:freespacedecomp}
\end{figure}

The following somewhat surprising statement is the main result of this section and shows that the slope information on the bottommost and leftmost edges of the original cells of $\freespacediagram_{\delta}$ leads to a construction recipe for the combinatorial structure of $G_{\delta}$.
\begin{lemma}\label{lem:reconstructfromextrema}
Assume $\delta$ is such that $\sings=\emptyset$. There is an algorithm that reproduces the combinatorial structure of $G_{\delta}$, using the sets $E_h^+,E_h^-, E_v^+,E_v^-$, and $\intersections$ along with the slope information on the bottommost and leftmost edges of $\freespacediagram_{\delta}$, in time $O(\lvert \intersections\rvert)$.
\end{lemma}
\begin{proof}
By Lemma~\ref{lem:montonearcs}, each subcell $\subcell$ contains only arcs that are monotone in both coordinates.  
We show that the slope information at the points in $\intersections_{\subcell}$ that lie on the lower and left edge of $\subcell$ is sufficient to know how they are connected through arcs in $B_{\delta}$ inside $\subcell$. Transferring the slope information along arcs, we can thus find the structure of $G_{\delta}$ in all subcells incrementally, starting from the bottom left, where after every step there is a (new) subcell in $\{\subcell\}$ where the slope information is known on the bottom and left edges. 

Assume first for simplicity that there are no points on the corners of $\subcell$. We remove all points in $\intersections_{\subcell}$ (corresponding to points in $E_v$ or $E_h$) 
that only touch $\subcell$, so that $B_{\delta}$ does not enter $\subcell$.  Each point $z\in \intersections_{\subcell}$ then has an incident arc that enters the interior of $\subcell$, so every point in $\intersections_{\subcell}$ has to be paired with another such point. To find how they are matched, the algorithm proceeds incrementally, deducing at least one set of matching points at every step. 

We consider a number of cases, depending on the behavior of $B_{\delta}$ at the as-of-yet unmatched points closest to the corners. Denote by $e_b,e_r,e_t$ and $e_l$ the bottom, right, top, and left edge of $\subcell$. We further write $z_l^b$ and $z_r^b$ for the two unmatched points in $\intersections_{\subcell}$ on $e_b$ that are farthest apart, i.e., for the left-most and right-most unmatched points, respectively. We similarly define the points $z_b^l,z_t^l,z_l^t,z_r^t,z_t^r$ and $z_b^r$ on the edges $e_l,e_t,$ and $e_r$, as illustrated in Figures~\ref{fig:casetype1} and~\ref{fig:examplescases}. The combinations of the slope information at the points $z_l^b,z_r^b, z_b^l$, and $z_t^l$ on the left and bottom edges yields a total of $2^4=16$ cases to consider, where we first assume that these 4 points are distinct. We will show how every case leads to a unique insertion of an arc of $B_{\delta}$ matching at least two vertices of $G_{\delta}$, at which point we restart the process to insert another arc until every point is matched.

\subparagraph*{Case type 1}\label{subsec:case1}

Consider the case where the arcs of $B_{\delta}$ at $z_l^b$ and $z_r^b$ move toward each other as in Figure~\ref{fig:casetype1}(left). Then there is no arc in $B_{\delta}$ that connects any point on $e_l$ to $e_r$ and thus $z_t^l$ necessarily connects to $z_l^t$. Since the same argument holds if the arcs of $B_{\delta}$ at $z_t^l$ and $z_b^l$ curve toward each other as in Figure~\ref{fig:casetype1}(right), this resolves a total of $7$ different possibilities. 

\begin{figure}[t]
\centering
  \begin{subfigure}[t]{0.49\textwidth}
  \centering
   \includegraphics[page=3]{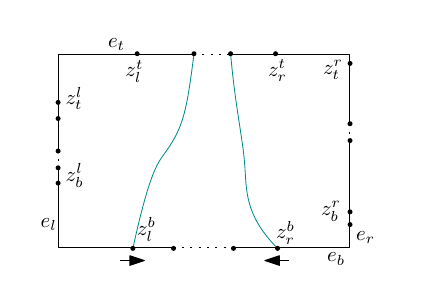}
  \end{subfigure}
    \begin{subfigure}[t]{0.49\textwidth}
  \centering
   \includegraphics[page=4]{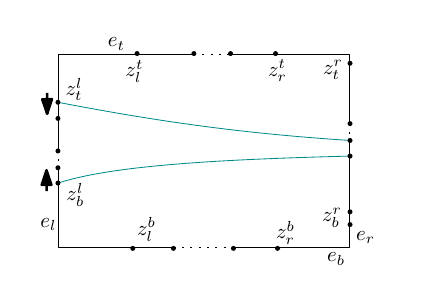}
  \end{subfigure}
      \caption{slope information - case type 1}\label{fig:casetype1}
\end{figure}

\subparagraph*{Case type 2}\label{subsec:case2}
Let $B_{\delta}$ be decreasing at $z_l^b$ and at $z_b^l$. Since each arc is monotone, $z_l^b$ cannot be joined to any point lying above $z_b^l$, so these two points have to be matched. This takes care of $4$ further cases.  

\begin{figure}[t]
\centering
  \begin{subfigure}[t]{0.49\textwidth}
  \centering
   \includegraphics[page=2]{freespacecasesz.pdf}
  \end{subfigure}
 \begin{subfigure}[t]{0.49\textwidth}
 \centering
    \includegraphics[page=1]{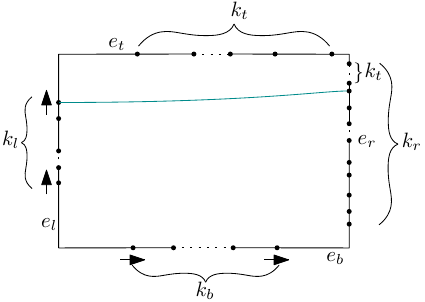}
  \end{subfigure}
      \caption{Examples of different cases.}\label{fig:examplescases}
\end{figure}
Consider next the case where $B_{\delta}$ is decreasing at $z_l^b$. The case of $B_{\delta}$ decreasing at $z_b^l$ having already been treated above, assume otherwise, so that $z_l^b$ does not connect to a point on $e_l$. Moreover, none of the points on $e_l$ can connect to points on $e_r$ or $e_b$, so they have to be joined to points on $e_t$. Therefore, if $k_l$ is the number of unmatched points on $e_l$, $z_l^b$ connects to the $(k_l+1)$-th point on $e_t$ (and the leftmost $k_l$ points on $e_t$ get joined to points on $e_l$). The case where $B_{\delta}$ decreases at $z_b^l$ is analogous, concluding a total of 15 cases. 

\subparagraph*{Last case}\label{subsec:case3}
The only remaining case not yet treated is the case where $B_{\delta}$ is increasing at each of the four points in question, as depicted in Figure~\ref{fig:examplescases}(right). There are two possibilities. If $z_t^l$ is not joined by an arc to $z_l^t$, then it must be joined to the $(k_t+1)$-th point $z_{k_t+1}^r$ on $e_r$ (counted from the top), where $k_t$ is the number of unmatched points in $\intersections_{\subcell}$ on $e_t$. 

The two cases can be distinguished by a simple rule. Let $k_b$ and $k_r$ denote the number of unmatched points on $e_b$ and $e_r$, respectively. 
If $z_t^l$ is connected to the $(k_t+1)$-th point on $e_r$, then all of the unmatched points in $\intersections_{\subcell}\cap e_t$ have to be connected to $e_r$ and no point on $e_l$ can be joined to $e_t$. Since no point on $e_l$ connects to $e_b$ because of the slope, we find that 
\begin{align}\label{eq:simplerule}
k_l+k_b+k_t=k_r.
\end{align}
Since every point has to be joined,~\eqref{eq:simplerule} alone also implies that no point on $e_l$ can be joined to $e_t$, distinguishing the two cases.

After having matched all points in $\intersections_{\subcell}$, we can transfer the slope information from points on the bottom and left edges to the points on $e_t$ and $e_r$. Note for this that the behavior does not change for points that are in neither of $E_h^+,E_h^-, E_v^+,$, or $E_v^-$, 
for which the slope at the corresponding points is clear. 
Lastly, the remaining points on $e_r$ and $e_t$ without slope information that have been matched to each other are marked as decreasing.

To conclude the proof, we next deal with degenerate cases and corner points. 
\subparagraph*{The degenerate cases}
Consider now the cases where there are no four pairwise distinct points $z_l^b,z_r^b, z_b^l$, and $z_t^l$. Notice first that the above analysis already treats the case where there are only three distinct such unmatched points. Also, none of the above cases require the slope information for all four points, so the case study also applies in this situation. 

The next case is that of only two distinct extremal points on $e_b\cup e_l$, which again can be treated very similarly. Likewise, the case of a single point on $e_b\cup e_l$ can be settled in the same way. Last, consider the undiscussed situation where all points on $e_b$ and $e_l$ have been matched. In this case, the unmatched points on $e_t$ have to be matched with those of $e_r$.
 
\subparagraph*{Corner points}
To conclude the proof that the points in $\intersections_{\subcell}$ can be uniquely joined in $\subcell$, we still need to discuss the case where there are points in $\intersections_{\subcell}$ on the corners of $\subcell$. For corner points on $e_b$ or $e_l$, the slope information tells us whether or not to include these points in $\intersections_{\subcell}$ for the matching process. 
To check if a point in the upper right hand corner $c_{tr}:=e_t\cap e_r$ needs to be considered in the absence of slope information for that corner point, we first observe that none of the points on $e_t$ or $e_r$ can connect to $c_{tr}$. 
If $c_{tr}$ connects to some other point in $\intersections_{\subcell}$, then no point on $e_r$ can connect to $e_t$. Hence, we see that  
\begin{align}\label{eq:simplerule2}
k_l+k_b-2n_{bl}= k_t+k_r,
\end{align}
where $n_{bl}$ denotes the number of points on $e_b$ that connect to a point on $e_l$. The quantities $k_l,k_b,k_t,k_r$ are defined as above, 
with the convention that the top left and bottom left corner only counts for $k_l$, the bottom right corner for $k_b$, and $c_{tr}$ for $k_t$. Note that the points in any pair of matched points in $e_l\cup e_b$ are incident to a decreasing arc of $B_{\delta}$. Therefore, none of them can connect to $c_{tr}$. Furthermore, no such pair of points can be separated by an arc from $c_{tr}$. Thus, \eqref{eq:simplerule2} can only be valid for a specific value of $k_t$.

In case $c_{tr}$ is not part of $\intersections_{\subcell}$, let $n_{tr}\ge 0$ denote the number of points on $e_t$ that are connected to a point on $e_r$. Then $
k_l+k_b-2n_{bl}= (k_t-1)+k_r-2n_{tr}$,
and we see that the right hand side decreases with each pair of matched points from $e_t$ and $e_r$. All in all, we see that to check if $c_{tr}$ belongs to $\intersections_{\subcell}$, one calculates the left and right hand side of \eqref{eq:simplerule2}. Their difference determines if $c_{tr}$ belongs to $\intersections_{\subcell}$.
\end{proof}

\begin{remark}
Figure~\ref{fig:derivatives} illustrates the necessity of some knowledge of the slope information on edges of a subcell for the accurate reconstruction of $B_{\delta}$ inside a cell. 
\begin{figure}[t]
  \centering
   \includegraphics[width=0.9\textwidth]{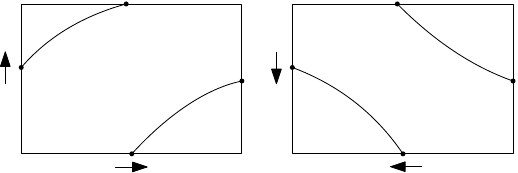}
      \caption{Two different sets of slope information and their combinatorial structures in a subcell.}\label{fig:derivatives}
\end{figure}
\end{remark}
\subparagraph*{Deriving the slope information} 
 In the following, we explain how to derive the slope information for points in a subcell by evaluating the distance between the two curve segments at specific \textbf{test points} along their parametrization. The process is illustrated in Figure~\ref{fig:testpoints}.  
We explain the approach using the example of an extremal point $z\in E_h$, which belongs to either $E_h^+$ or $E_h^-$, depending on the distance evaluated at three test points. The three points are any points lying, sufficiently close, below, above, and to the left (or right) of $z$. Two of the three points either both belong to the free space or both to the forbidden region, and one will not, determining whether $z\in E_h^+$ or $z\in E_h^-$, as illustrated in Figure~\ref{fig:testpoints}. A similar analysis can be used to determine the slope information associated to other points $z\in I$ (as long as $z\not\in \sings$).  

Since $z$ lies on the (vertical) boundary of a cell, at least two of the points have natural candidates: We can choose any point on the boundary of the cells that is closer to $z$ than any other intersection point with $B_{\delta}$ lying on the same (vertical) line. 
To guarantee that the chosen point in the horizontal direction is sufficiently close to $z$ for the analysis to work, we use the same machinery as before in the construction of the cell decomposition. Denoting $z=(z_x,z_y)$, let $z^*$ be the first number smaller than $z_x$ where $\curve_1$ (parametrized by the horizontal axis) intersects the ball of radius $\delta$ centered at $\curve_2(z_y)$. Then the point of evaluation is chosen from the open interval $(z^*,z_x)$.

\begin{figure}[t]
  \centering
   \includegraphics{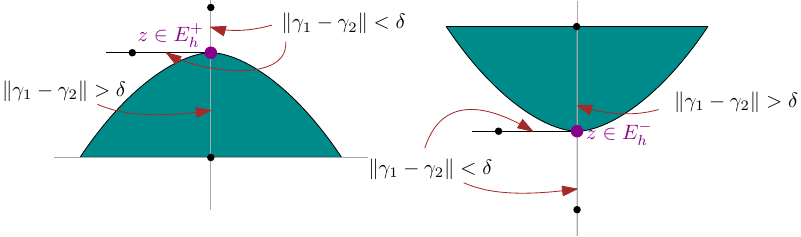}\\
   \includegraphics{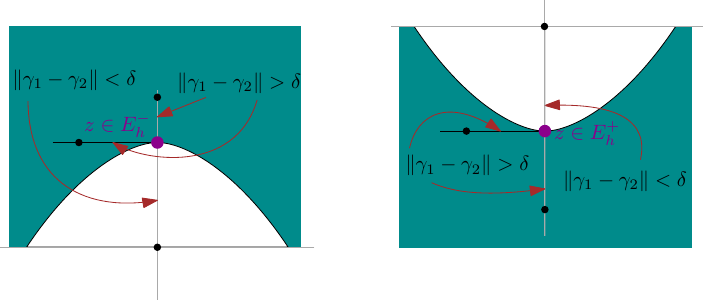}
      \caption{Finding slope information by evaluating the distance at points.}\label{fig:testpoints}
\end{figure}

\subsection{Monotone paths in the free space diagram and the decision problem}\label{sec:monotonealgo}

In this section, we use the machinery developed above to derive an algorithm that answers the decision problem. 
\begin{definition}\label{def:relativefreespace}
The \textbf{reachable free space} $\rfreespace_{\le\delta}(\curve_1,\curve_2)$ of two curves $\curve_1$ and $\curve_2$ is the subset of all points of the free space $\freespace_{\delta}(\curve_1,\curve_2)$ that are reachable from the origin by a path monotone in both coordinates.
The \textbf{complexity} $N_{\le \delta}(\curve_1,\curve_2)$ of the reachable free space is the number of original cells with non-empty intersection with $\rfreespace_{\le\delta}(\curve_1,\curve_2)$. 
\end{definition}
As before, we can assume (by Proposition~\ref{prop:algebraicregvalues} below) that $\delta$ is such that there are no singularities on the boundary curves $B_{\delta}$ of the free space. 
We will later show (Proposition~\ref{prop:finitepointsfreespace}) that for algebraically bounded curves the decomposition of $\freespacediagram_{\delta}$ and point set in Lemma~\ref{lem:reconstructfromextrema} (corresponding to the vertices of the graph $G_{\delta}$) consists of $O(mn)$ elements (depending on the allowed degrees of the curves). 
To obtain bounds on the running time of the algorithm, in the following we assume that the investigated curves are algebraically bounded.
A \textbf{reachable interval} $R$ is a maximal interval of points on a boundary edge of a subcell, such that every point in $R$ is reachable from the origin by a path contained in the free space, that is monotone in both coordinates. In particular, horizontal (vertical) reachable intervals always end on the right (top) either at a vertex of $G_{\delta}$, or at a cell wall, so there are at most $O(mn)$ reachable intervals. Consider, for example, a reachable interval $R$ on the bottom edge of a cell and its left-most point $b_l$. Points reachable from $b_l$ clearly include all points reachable from any point to the right of $b_l$ within the same interval of free space on the edge. 
By Lemma~\ref{lem:montonearcs}, there is at most one reachable interval on each of the top and right edges of a subcell $\subcell$ that is reachable by a monotone path starting from a reachable interval on the left (or bottom) edge of $\subcell$, as illustrated in Figure~\ref{fig:reachableintervals}. The set of reachable intervals on the top and right edges of $\subcell$ can be readily computed in constant time from the coordinates of the vertices of $G_{\delta}$ along with its combinatorial information. 

\begin{figure}[t]
  \centering
   \includegraphics[width=0.9\textwidth]{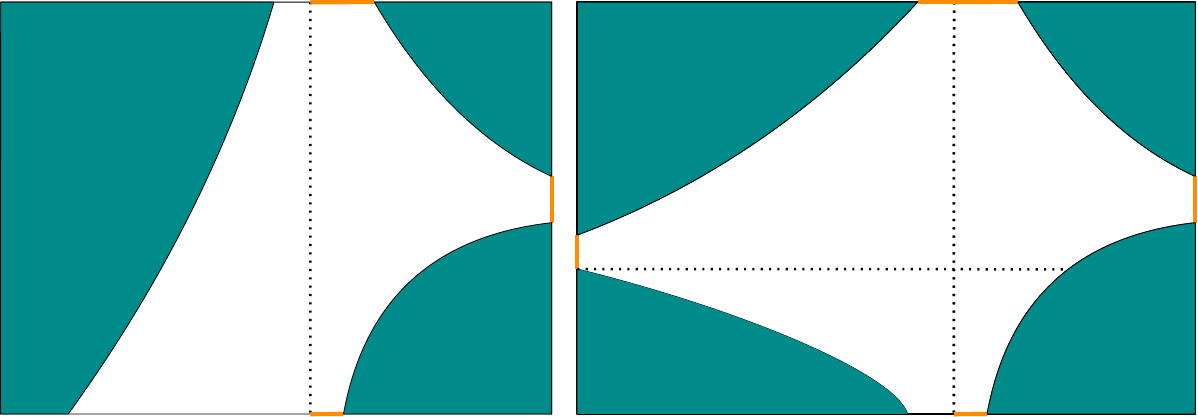}\\
      \caption{Computing the reachable intervals (in orange) on the top and right edges of a cell.}\label{fig:reachableintervals}
\end{figure}

We obtain a graph representation of the reachable intervals as follows. Every reachable interval corresponds to a vertex and two vertices are joined by an edge if both intervals belong to the same subcell and every point of one of the reachable intervals can be joined by a monotone path to some point on the other interval. To construct the resulting graph $\mathcal{R}$, consider the original cell decomposition of $\freespacediagram_{\delta}$.
We start by checking if $(0,0)\in \freespace_{\delta}$ and then iterate through the original cells from the bottom left in every row and subsequently move up row by row, computing the reachable intervals on the top and right boundary edges of each cell at every step. To compute these reachable intervals, we compute the reachable intervals on the constant number of subcells in each cell in the same fashion, processing a total of $O(N_{\le \delta}(\curve_1,\curve_2))$ cells. Observe that all reachable intervals have $x$ and $y$ coordinates that are each already present in the coordinates of the marked points, i.e., vertices of $G_{\delta}$. (We can also compute the combinatorial structure of $G_{\delta}$ within a subcell while processing it, as described in the proof of Lemma~\ref{lem:reconstructfromextrema}.) All in all, we obtain the following result.

\begin{proposition}\label{prop:decisionproblempath}
Given two algebraically bounded piecewise smooth curves $\curve_1$, $\curve_2$ in $\RR^d$ comprised of $m$ and $n$ pieces, respectively, and a value of $\delta$ such that $B_{\delta}$ has no singularities, one can decide if $\distf(\curve_1,\curve_2)\le\delta$. The running time is bounded by $O(mn)$.
\end{proposition}
\begin{remark}
In contrast to the situation for polygonal curves, our solution to the decision problem does not automatically yield a parametrization of $\curve_1$ and $\curve_2$ that realizes the given distance of $\delta$. One way of obtaining such parametrizations would be through parametrizations of the arcs of $B_{\delta}$.
\end{remark}
We note that a necessary ingredient for the above decision algorithm is being able to compare the coordinates of the points in the set of points used in the construction of the partitioning of $\freespacediagram_{\delta}$.  No further algebraic operations are required to solve the decision problem. 
Observe that the answer to the decision problem depends only on the combinatorial structure of $G_{\delta}$ and the ordering of the $x$- and $y$-coordinates of the marked points $\intersections$ that give rise to the partition. Away from critical values of $\delta$, the answer therefore does not change when applying a sufficiently small perturbation of the marked points. Furthermore, in the following section, we deduce that the number of marked points for algebraically bounded curves is bounded, and also show that the number of values of $\delta$ for which the answer to the decision problem is subject to change is bounded for such curves. In particular, we observe that the decision algorithm does not crucially depend on evaluating the required algebraic operations exactly, but can also be carried out approximately.

\section{The structure of the boundary $B_{\delta}$ of the free space}\label{sec:singularitystructure}
In this section, we complement the introduced computational approach with results concerning the structure of singular points on $B_{\delta}$ that guarantee that our algorithms work correctly. The material in this section provides the technical foundation for the approach outlined above. 
Observe that even for straight edges, in case of the $\ell_1$ or $\ell_\infty$ norm, $B_{\delta}$ can feature non-smooth points that persist when varying $\delta$, as the free space is the intersection of a parallelogram with $[0,1]^2$ (refer to~\cite{ALT1995}). Although our framework can be adjusted to include these cases, for simplicity, we exclude the values $p=1$ and $p=\infty$ from our treatment. We discuss the different scenarios for all remaining $\ell_p$ norms, starting with general smooth curves and subsequently focusing on algebraically bounded curves; see~\cite{Hartshorne1977} for background on algebraic curves and degree. The central property of an algebraic curve we use is that it cannot intersect a hyperplane in general position more than a constant number (depending on the allowed degree) of times. 
\begin{figure}[t]
  \centering
   \includegraphics{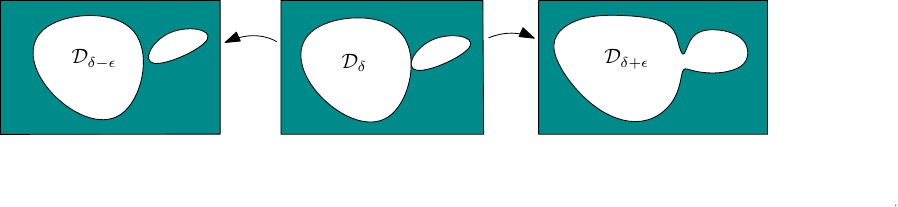}
      \caption{Illustration of singular points and changes of the boundary $B_{\delta}$ as $\delta$ changes.}\label{fig:freespaceboundary}
\end{figure}

Defining the function $f_{p,\curve_1,\curve_2}:[0,1]^2\to \RR$ by $f_{p,\curve_1,\curve_2}(t_1,t_2)=\sum_i(|(\curve_1)_i(t_1)-(\curve_2)_i(t_2)|)^p$, we have that $B_\delta=\{(t_1,t_2) |f_{p,\curve_1,\curve_2}(t_1,t_2)=\delta^p\}\subset [0,1]^2$. There are two possible types of singular points of $B_{\delta}$ --- cusps and points of self-intersection --- characterized as those points that do not have a well-defined tangent defined by the Jacobi matrix (gradient) $J:=Df_{p,\curve_1,\curve_2}$. The singular points correspond to \textbf{critical points} of the function $f_{p,\curve_1,\curve_2}$. The image of a critical point is known as a \textbf{critical value}. The entries $J_k(t_1,t_2)$ for $k=1,2$ are given by    
\begin{align}\label{eq:jacobian}
J_{k}(t_1,t_2)=p\sum_i|(\curve_1)_i(t_1)-(\curve_2)_i(t_2)|^{p-2}((\curve_1)_i(t_1)-(\curve_2)_i(t_2))  (-1)^{k+1}(\curve_k)'_i(t_k).
\end{align}
The singular points of $B_\delta$ can thus be described as 
\begin{equation}\label{eq:boundarysingularities}
\begin{aligned}
\sings=\biggl\{(t_1,t_2) \bigg|& \sum_i(|(\curve_1)_i(t_1)-(\curve_2)_i(t_2)|)^{p-2}((\curve_1)_i(t_1)-(\curve_2)_i(t_2))  (\curve_1)'_i(t_1)=0,\\
-&\sum_i(|(\curve_1)_i(t_1)-(\curve_2)_i(t_2)|)^{p-2}((\curve_1)_i(t_1)-(\curve_2)_i(t_2))  (\curve_2)'_i(t_2)=0\biggr\}\cap B_\delta.
\end{aligned}
\end{equation}
Geometrically, critical points of $f_{p,\curve_1,\curve_2}$ for $1<p<\infty$ can be described as pairs of points $(t_1,t_2)$ on $B_\delta$ satisfying the additional criteria
\begin{align}\label{eq:boundarysingularitiesgeom}
\frac{\D}{\D t_1}\|\curve_1(t_1)-\curve_2(t_2)\|_p=0=\frac{\D}{\D t_2}\|\curve_1(t_1)-\curve_2(t_2)\|_p.
\end{align}

Observe that both equations in~\eqref{eq:boundarysingularities} (or~\eqref{eq:boundarysingularitiesgeom}) are continuous for $1< p<\infty$. The implicit function theorem guarantees that, away from critical points, $B_\delta$ can be described as a collection of planar $C^1$ curves, as illustrated in Figure~\ref{fig:freespaceboundary}.

The intuitive idea behind the investigations of this section is that since the two equations in~\eqref{eq:boundarysingularities} are independent of $\delta$, the solution set (the singular points) should ideally be a collection of isolated points, so that there are only isolated critical values of $\delta$ to which they belong. 

Since both $\curve_1$ and $\curve_2$ are $C^2$, the function $f_{p,\curve_1,\curve_2}$ is $C^2$ for $2\le p< \infty$, and $C^1$ for $1<p<2$ (see also~\eqref{eq:hessianboundarysingularities} below). 
For the sake of completeness, we derive the coefficients $H_{kl}(t_1,t_2)$, with $k,l\in \{1,2\}$, of the Hessian $H$ of $f_{p,\curve_1,\curve_2}$, from the Jacobian $J_k$ in~\eqref{eq:jacobian} of $f_{p,\curve_1,\curve_2}$. Up to multiplication by a constant $C$, we compute $H_{kl}$ by differentiating the equations~\eqref{eq:boundarysingularities}, which yields
\begin{equation}\label{eq:hessianboundarysingularities}
\begin{aligned}
C\cdot H_{kl}(t_1,t_2)=&\sum\limits_{i}|(\curve_1)_i(t_1)-(\curve_2)_i(t_2)|^{p-2}\biggl((p-2)(-1)^{l+k}(\curve_l)'_i(t_l)(\curve_k)'_i(t_k)+\\
&(-1)^{l+k}(\curve_k)'_i(t_k)(\curve_l)'_i(t_l)+(-1)^{k+1}\delta_{kl}((\curve_1)_i(t_1)-(\curve_2)_i(t_2))(\curve_k)''_i(t_k) \biggr)\\
=&\sum\limits_i |(\curve_1)_i(t_1)-(\curve_2)_i(t_2)|^{p-2}\biggl( (-1)^{k+l}(p-1)(\curve_l)'_i(t_l)(\curve_k)'_i(t_k)\\
&+(-1)^{k+1}\delta_{kl}((\curve_1)_i(t_1)-(\curve_2)_i(t_2))(\curve_k)''_i(t_k)\biggr),
\end{aligned}
\end{equation}
where $\delta_{kl}$ is the Kronecker delta, which has the value one for $k=l$ and zero otherwise.

\begin{example}\label{ex:infinitelysings}
In case of the Euclidean norm, where $p=2$, the singularities $(t_1,t_2)$ satisfy the orthogonality condition 
\begin{align}\label{eq:orthocrit}
\braket{\curve_1(t_1)-\curve_2(t_2),\curve_1'(t_1)}=0=\braket{\curve_1(t_1)-\curve_2(t_2),\curve_2'(t_2)}.
\end{align}
The singularities thus correspond to pairs of points, one on each of the curves, at which the two tangent vectors are contained in the hyperplane orthogonal to the vector between the pair of points.

Consider the two curves $\curve_1(t)=(t,0)$ and $\curve_2(t)=(t,t^5\sin(\frac{1}{t})+\alpha)$ for $t\neq 0$ and $\curve_2(t)=(0,\alpha)$, for some $\alpha\in \RR$. Notice that both $\curve_1$ and $\curve_2$ are $C^2$ curves. Since $(\curve_2)_2$ has infinitely many zeros in $[0,1]$, its derivative also has infinitely many zeros in $[0,1]$, so there are infinitely many critical points of $f_{p,\curve_1,\curve_2}$ in $[0,1]$ for $p=2$, giving rise to infinitely many critical values. Observe also that the related function $\curve_2(t)=(t,\exp(-\frac{1}{t^2})\sin(\frac{1}{t})+\alpha)$ for $t\neq 0$ and $\curve_2(0)=(0,\alpha)$ is infinitely differentiable in $[0,1]$ and has the same property.
\end{example} 
\begin{remark}
One might find it desirable to find assumptions which imply that the function $f_{p,\curve_1,\curve_2}$ is a Morse function. A \textbf{Morse function} is a $C^2$ function where the Hessian is non-singular at critical points, see~\cite{Guillemin2010}. The Morse lemma would then imply that the critical points described by the equations~\eqref{eq:boundarysingularities} are isolated, so that the same is also true for the critical values. However, it is noteworthy that there are cases where a small rotation and translation of one of the curves is not be enough to ensure that the critical points are isolated, as can be seen by an adjustment of Example~\ref{ex:infinitelysings}.
\end{remark}
We have the following result regarding the \textbf{regular}, i.e., non-critical values of $f_{p,\curve_1,\curve_2}$. To simplify the exposition, we assume that $\curve_1$ and $\curve_2$ are smooth, so that $\freespacediagram_{\delta}$ consists of a single cell. In particular, we can ignore singularities of $B_{\delta}$ due to the nonsmoothness of the curves where the smooth pieces join. 
\begin{proposition}\label{prop:regvalues}
The set $\regs$ of regular values of $f_{p,\gamma_1,\gamma_2}$ is open in $\RR$ for $1< p<\infty$ and dense if $2\le p<\infty$. 
For $1< p< 2$, there is a dense and open set in $\RR$ in which the only possible singular points of $f_{p,\curve_1,\curve_2}$ are the points $(t_1,t_2)$ where $(\curve_1)_i(t_1)=(\curve_1)_i(t_2)$ for some $1\le i\le d$. 
If the points $(t_1,t_2)$ where $(\curve_1)_i'(t_1)=(\curve_2)_i'(t_2)=0$ for all $i$ form a null set, then $\regs$ is also dense for $1<p<\infty$.    
\end{proposition}
\begin{proof}
Let $\delta_n\xrightarrow{n\to\infty}\delta$ be a sequence of critical values with limit $\delta$. For every $\delta_n$, there is a singular point $(t_1^n,t_2^n)$ satisfying~\eqref{eq:boundarysingularities}. Since $[0,1]^2$ is compact, we can assume that $(t_1^n,t_2^n)$ converges by passing over to a subsequence. The set of critical values is then closed by continuity of~\eqref{eq:boundarysingularities} and $f_{p,\curve_1,\curve_2}$ in both $t_1$ and $t_2$ for $1<p<\infty$.

For $2\le p<\infty$, $f_{p,\gamma_1,\gamma_2}$ is $C^2$, so by the classical Morse-Sard theorem~\cite{Sard1942} the set of those $\delta^p>0$ that are the image of singular points in $B_{\delta}$ has Lebesgue measure zero in $\RR$. 
As the complement of a zero-set for the Lebesgue measure, the set of regular values is dense. 

We argue similarly for $1<p<2$. The regularity of $f_{p,\curve_1,\curve_2}$ is insufficient to obtain strong bounds on the size of the set of singular values, even with recent generalizations of the Morse-Sard theorem~\cite{Ferone2020}. We thus restrict attention to points $(t_1,t_2)$ where $(\curve_1)_i(t_1)\neq(\curve_1)_i(t_2)$ for all $i$, where the Jacobian~\eqref{eq:jacobian} of $f_{p,\gamma_1,\gamma_2}$ is locally Lipschitz, which implies that the set of singular values has measure zero, at least when restricted to these points~\cite[Theorem~1]{Bates1993}. The critical points $(t_1,t_2)$ of the equation $(\curve_1)_i(t_1)-(\curve_2)_i(t_2)=0$ are given by $(\curve_1)'_i(t_1)=0=(\curve_2)'_i(t_2)$. If these form a set of measure zero in $[0,1]^2$ for each $i$, then the image of their union under the $C^1$ map $f_{p,\curve_1,\curve_2}$ is a null set with respect to the Lebesgue measure.
\end{proof}
\begin{remark}
If the curves $(\curve_1)_{i}$ and $(\curve_2)_i$ are both Morse functions for all $i$, which is almost always the case, then $(\curve_1)_{i}(t_1)-(\curve_2)_i(t_2)$ also defines a Morse function $[0,1]^2\to \RR$ on the product space $[0,1]^2$~\cite[Exercise 4]{Audin2014}, which implies, by the Morse lemma, that the critical points are isolated~\cite[Theorem~3.1.1]{Nirenberg2001}. Therefore, the last assumption in Proposition~\ref{prop:regvalues} ensuring that $\regs$ is dense for $1<p<\infty$ is satisfied for almost all curves. 
\end{remark}

Proposition~\ref{prop:regvalues} does not rule out accumulation points of singular values for general smooth curves, even when $p=2$, as illustrated in Example~\ref{ex:infinitelysings}.

\begin{lemma}\label{lem:nocusps}
Let $1<p<\infty$, $c$ be a simple component of $B_{\delta}$, such that $[0,1]^2\setminus c$ has two components, each with non-empty interior. Let $R_c$ be a region in $[0,1]^2$ enclosed by $c$. If $\overline{R_c}$ admits a neighborhood that has empty intersection with all other components $c'\neq c$ of $B_{\delta}$, then $c$ is everywhere differentiable, i.e., $B_{\delta}$ does not have any cusps.
\end{lemma}

\begin{proof}
We sketch the proof, which can be seen as a consequence of the classification of the shape of the free space inside a cell of the free space diagram associated to two line segments~\cite[Lemma~3]{ALT1995}, as the intersection of a smooth ellipse with a rectangle. The statement there is only for the $\ell_2$-norm, but can be naturally extended by noting that for $1<p<\infty$, the boundary of the set of all points with norm $\le \delta$ from a given point is $C^1$, as follows from the regular value theorem since any $\delta>0$ is a regular value of the $\ell_p$-norm. 

The question of whether a point $(t_1,t_2)$ on $c$ corresponding to the two points $\curve_1(t_1)$ and $\curve_2(t_2)$ on the two curves with distance $\delta$ is a cusp is inherently a local one. By differentiability of the two curves, they are approximated locally by their tangents and within a sufficiently small neighborhoods of $t_1$ and $t_2$, we can picture the curves as a result of a differentiable deformation of their tangents at these points. Such a deformation leads to a similar local deformation of the boundary of the free space. Indeed,  the boundary of the free space for the linear tangents is $C^1$ for noncritical values of the distance parameter $\delta$, which remains that way after a small local differentiable deformation of the linear tangents. We therefore have $C^1$ regularity of $c$ at $(t_1,t_2)$. 
\end{proof}
As a consequence of Lemma~\ref{lem:nocusps}, the critical values of $f_{p,\gamma_1,\gamma_2}$ for $1<p<\infty$ coincide with the appearance of a new component, or the merging of components in the free space diagram.

We will show that if the piecewise smooth curves $\curve_1$ and $\curve_2$ consist of algebraically bounded curves, Proposition~\ref{prop:regvalues} can be strengthened. To this end, we need some preparation.
\begin{lemma}\label{lem:finiteradialextrema}
Let $\curve$ be an algebraically bounded curve. Then $\curve$ is tangent to only finitely many $\ell_2$-spheres centered at the origin, given by $\{x\in\RR^d\;\mid\; \|x\|_2=r\}$ for some $r\in \RR_+$, with bound depending only on the degree bound of the curve.
\end{lemma}

\begin{proof}
First note that the statement of the lemma is well-known to be true if we replace $\ell_2$-spheres by hyperplanes orthogonal to a coordinate axis. Indeed, an algebraic curve of a given degree only has finitely many extrema in a given coordinate direction (if it is not contained in such a plane), as the extrema correspond to algebraic sets themselves, obtained through partial derivatives of each of the defining equations of $\curve$. To derive the statement for $\ell_2$-spheres, we make use of a variant of stereographic projection that considers all $\ell_2$-spheres of positive radius at the same time.  For a sphere $S^{d-1}_{R^2}$, centered at the origin, with radius $R^2$, we define the function (see Figure~\ref{fig:stereoprojection})

\[
S^{d-1}_{R^2}\owns (x_1,\ldots,x_d)\mapsto \left(\frac{2R^2}{R^2-x_d}(x_1,\ldots,x_{d-1}),-R^2\right),
\]   
with inverse 
\[
\RR^d\owns (x_1,\ldots,x_{d-1},-R^2)\mapsto \left(\frac{4R^4}{\sum_{i=1}^{d-1}x_i^2+4R^4} (x_1,\ldots,x_{d-1}),\left(1-2\frac{4R^4}{\sum_{i=1}^{d-1}x_i^2+4R^4}\right)R^2\right).
\]
\begin{figure}[t]
  \centering
   \includegraphics{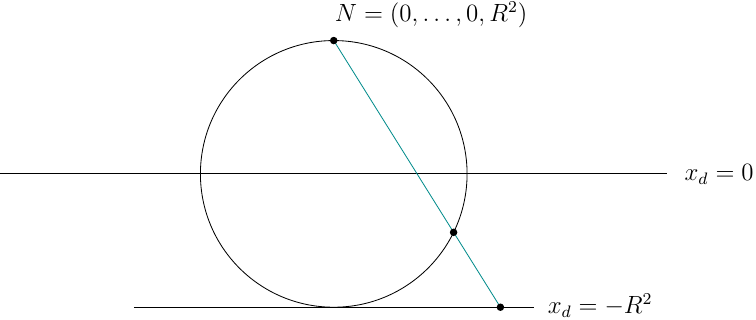}
      \caption{Illustration of the variant of stereographic projection we use in Lemma~\ref{lem:finiteradialextrema}.}\label{fig:stereoprojection}
\end{figure}
Outside of the points where $x_1=x_2=\ldots=x_{d-1}=0$, we obtain a rational map $\Psi$ defined on all of $\RR^d$, mapping all $\ell_2$-spheres centered at the origin to hyperplanes, with rational inverse. Therefore, if $\curve$ was tangent to infinitely many spheres, the algebraic curve that is its image under $\Psi$ would be tangent to infinitely many hyperplanes. However, as pointed out above, the number of times this can happen is a finite number depending on the degree of the algebraic curve in question, concluding the proof.  
\end{proof}

\begin{lemma}\label{lem:finiteintersectionsphere}
Algebraically bounded curves can intersect an $\ell_p$-sphere only a finite number of times, with bound depending only on the degree bound.
\end{lemma}
\begin{proof}
Assume otherwise, so that there is a sphere $S_p$ that intersects a curve $\curve$ some large number of times. Since $S_p$ is compact, there are arbitrarily many intersection points within a small neighborhood $U$ in $S_p$. Consider the pencil of $\ell_2$-spheres centered at a point in the middle of $U$. One readily sees that the number of times $\curve$ is tangent to some sphere in this pencil can be made arbitrarily large, which would contradict Lemma~\ref{lem:finiteradialextrema}.
\end{proof}
The underlying idea of the proof of Lemma~\ref{lem:finiteintersectionsphere} can be also used to show the following.
\begin{proposition}\label{prop:algebraicregvalues}
Let $\curve_1$ and $\curve_2$ be algebraically bounded curves. Then there are only finitely many (with bound depending only on the degree bound) critical values of $f_{p,\curve_1,\curve_2}$ for $1<p<\infty$. 
\end{proposition}
It is instructive to first consider the important case that $p=2$, and both $\curve_1$ and $\curve_2$ have (componentwise) parametrizations as polynomials or more generally rational functions. Observe that in this case the set of critical points of $f_{p,\curve_1\curve_2}$ is an algebraic set, and thus the finite union of irreducible algebraic sets, which correspond to connected components in both the Zariski and the usual topology~\cite{Hartshorne1977} (for this, instead of just real numbers, we also allow complex numbers as input for $f_{p,\curve_1,\curve_2}$). The image on each connected set of critical points under the $C^1$ function $f_{p,\curve_1,\curve_2}$ is constant, so the image of all critical points leads to only finitely many critical values of $f_{p,\curve_1,\curve_2}$. Note that the number of irreducible algebraic sets is bounded in terms of the degree bound of the curves.
\begin{proof}[Proof of Proposition~\ref{prop:algebraicregvalues}]
We now treat the general case, for which we argue geometrically. By Lemma~\ref{lem:nocusps}, for $1<p<\infty$, a critical value $\delta$ of the function corresponds to either the appearance of a new component of $B_{\delta}$, or to two (or more) components becoming tangent, at some critical point $(x,y)\in [0,1]^2$. 
In $\RR^d$, by~\eqref{eq:boundarysingularitiesgeom}, these events correspond to $\curve_1$ becoming tangent at $\curve_1(x)$ to an $\ell_p$-ball centered at a point $\curve_2(y)$; conversely, the same ball centered at $\curve_1(x)$ is tangent to $\curve_2$ at $\curve_2(y)$. 

Assume that there are infinitely many distinct critical values $\delta$ where these events take place. Then, by compactness of $[0,1]^2$, there is a convergent subsequence of critical points in $[0,1]^2$ corresponding to pairs of points on the two curves where these events happen for distinct values of $\delta$. We can assume that the terms of the sequence are pairwise distinct. We consider the sequences as given by $(x_k)_{k\in\NN}$ and $(y_k)_{k\in\NN}$ in the first and second coordinates, respectively, and assume, without loss of generality, that the $y_k$ are pairwise distinct. The tangent vectors of $\curve_1$ at $\{x_k\}$ lie in the hyperplanes tangent to balls centered at the convergent sequence of points $\curve_2(y_k)$. For sufficiently large $k$, the tangent vectors of $\curve_1$ at $\curve_1(x_k)$ are roughly the same. As a consequence, the pencil of $\ell_2$-spheres based at a point $z$ sufficiently close to $\curve_1(x_k)$ contains an unbounded number of tangencies to the curve $\curve_2$ as $\curve_2$ alternates between moving away from and towards $z$. This is illustrated in Figure~\ref{fig:criticalsisolated}, showing the tangent vectors of $\curve_2$ `bundle' together around the eventual limit vector. We can further assume that the $\ell_2$ spheres do not contain an entire part of $\curve_2$ by applying a small perturbation to the sphere if necessary.
However, this would again contradict Lemma~\ref{lem:finiteradialextrema}, so we can discard the assumption of infinitely many critical values of $f_{p,\curve_1,\curve_2}$. 
\begin{figure}[t]
  \centering
   \includegraphics{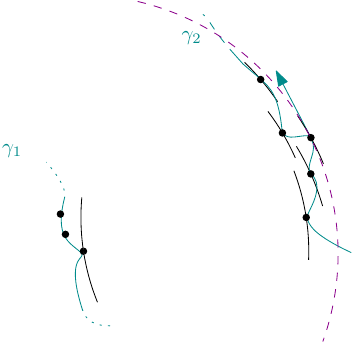}
      \caption{A section (the dashed magenta line) of a sphere of the pencil that contains a large number of tangencies to $\curve_2$.}\label{fig:criticalsisolated}
\end{figure}
To see that the number of critical values is in fact bounded by a constant depending on the bound on the degrees of the algebraic curves, we assume otherwise and argue very similarly to above. Indeed, given a sequence of curves with unbounded number of critical values leads to an arbitrary number of critical points within an arbitrary small neighborhood in $[0,1]^2$. By rescaling, we can assume that the arclength of the curves in the sequence is bounded and thus, just as above, we have that an arbitrary number of critical points correspond to points and tangents on each of the curves that are arbitrarily close to each other.    
\end{proof}


Similarly to the proof of Proposition~\ref{prop:algebraicregvalues} which bounds the number of critical values of $\delta$, we also obtain a similar bound on the number of marked points for a regular value of $\delta$.
\begin{proposition}\label{prop:finitepointsfreespace}
The number of points (components) in each of the sets $E_h^+,E_h^-, E_v^+,E_v^-$, (and $\intersections$) is at most a constant in each cell, for a total of $O(mn)$, whenever $\curve_1,\curve_2$ are algebraically bounded, consisting of $m$, resp. $n$ smooth pieces, and $\delta$ is a regular value. The resulting decomposition of the free space diagram consists of $O(mn)$ cells.
\end{proposition} 	
Proposition~\ref{prop:finitepointsfreespace} implies that the bound on the degree of a set of algebraically bounded curves determines a bound on the number of points used in the decomposition of $\freespacediagram_{\delta}$, for any regular value of $\delta$, thereby providing the justification for the complexity bounds on the running-time of our  decision algorithm. The proof is very similar to that of Proposition~\ref{prop:algebraicregvalues}.

\begin{proof}[Proof of Proposition~\ref{prop:finitepointsfreespace}]
Assume that there is a sequence of values for a given $\delta$ for which the number of points in, say, $E_h$ becomes greater than any given constant, then we can arrive at a contradiction, similarly to the proof of Proposition~\ref{prop:algebraicregvalues}.
  The idea is that in this case, there has to be an arbitrary number of points in, say, $E_h$, that are arbitrarily close to each other in $[0,1]^2$. By rescaling the curves (and $\delta$) to normalize the length of the longer curve, we can assume that these points correspond to arbitrarily close points along $\curve_1$ and in $\RR^d$. This again allows the construction of a pencil of $\ell_2$-spheres with an arbitrary number of tangencies to $\curve_2$, contradicting Lemma~\ref{lem:finiteradialextrema}.  

By Lemma~\ref{lem:finiteintersectionsphere}, we see that each of the finitely many extrema leads to only finitely many marked points. Again by rescaling, one can eliminate the possibility that the maximum number of points in $\intersections$ depends on $\delta$.  
\end{proof}

\section{The minimization problem of computing the Fréchet distance}\label{sec:minimalproblem}
The aim of this section is to use the solution of the decision problem to solve the optimization problem of finding the minimal $\delta$ for which the decision problem for two curves has a positive answer, which corresponds to the Fréchet distance between the curves. Assuming that all $m$, resp. $n$ smooth pieces of $\curve_1$, resp. $\curve_2$ are algebraically bounded curves, we show that the number of operations required for the computation is of the same complexity as the polygonal case. The results outlined in this section mostly essentially mirror those for smooth planar curves under the $\ell_2$-distance in~\cite{Rote2007}. 

The central idea is to invoke Megiddo's parametric search technique~\cite{Megiddo1983}, together with the optimization introduced by Cole~\cite{Cole1987}. Parametric search is the standard approach to computing the Fréchet distance and its variants~\cite{ALT1995,oostrum2004, Rote2007, chambers2010}. In the following, we explain the adjustments necessary to apply it to the situation of algebraically bounded curves. 

The combinatorial structure of the partition of the free space diagram introduced in Section~\ref{sec:combdescfreespace} changes only at certain \textbf{critical values} of $\delta$. In addition to these values being critical values of the function $f_{p,\curve_1,\curve_2}$ introduced in Section~\ref{sec:singularitystructure}, there are further possibilities. 
By Lemma~\ref{lem:nocusps}, the only values of $\delta$ where the decision problem can change are related to the emergence or merging of components of free space, a component becoming tangent to the boundary of walls of original cells, a marked point appearing or disappearing, or the ordering of the $x$- and $y$-coordinates of the marked points changing. The results of the previous section imply that the set of $\delta$ where none of these events occur is open. Moreover, we observe that a new marked point that is not naturally related to a previously marked point can appear (or disappear) only a finite number of times, bounded by a constant depending on the degree of the algebraically bounded curves. Similarly to the proof of Proposition~\ref{prop:finitepointsfreespace}, this is again a consequence of the fact that if there did exist a sequence of values of $\delta$ with an ever increasing number of such events, then there is a sequence of points on the curves that ultimately contradict the property that algebraically bounded curves cannot be tangent to sphere centered at a fixed point more than a fixed number of times.
A similar argument also shows that the ordering of the marked points cannot change more often than a constant depending on the degree bound of the curves, as the ordering is subject to change only when extremal points meet on the same vertical or horizontal line in $\freespacediagram_{\delta}$.

There are $O(mn)$ critical values of $\delta$ between which no marked points appear or disappear, and no components merge, appear, or start touching the boundary of cells. 
We apply a binary search among these $O(mn)$ critical values so that the only changes to the answer of the decision problem are due to a change of the order of the extremal points in the $x$- and $y$-direction. Every step of the binary search involves running the decision algorithm, resulting in a running time of $O(mn\log(mn))$ for this step. Note that since we cannot run our decision algorithm directly on these critical values, we have to move to a value for $\delta$ that lies between them. As a result, the binary search only narrows down the interval of $\delta$ to two possible, adjacent intervals, instead of one. Inbetween these $O(mn)$ critical values, we use Cole's variant of parametric search with a parallel sorting algorithm for both the $x$- and $y$-coordinates of the marked points of $B_{\delta}$, to obtain an overall running time of $O(mn\log(mn))$. The value of the Fréchet distance $\distf(\curve_1,\curve_2)$ is the lowest point of the lowest interval in which the decision problem has a positive answer. 
The result is summarized in Theorem~\ref{thm:smoothcomplexity}.
\begin{theorem}\label{thm:smoothcomplexity}
Let $\curve_1$ and $\curve_2$ be two algebraically bounded curves in $\RR^d$ consisting of $m$ and $n$ pieces, respectively. Then the Fréchet distance between $\curve_1$ and $\curve_2$ can be computed in $O(mn)$ space and in $O(mn\log(mn))$ operations (of bounded algebraic complexity).
\end{theorem}


\section{Simplification of piecewise smooth curves}\label{sec:simpcurves}
In this section, we introduce an algorithm to simplify a piecewise smooth curve consisting of $n$ smooth pieces and runs in $O(n)$ time. 
We adapt the approach of~\cite{Driemel2012} to our setting of piecewise smooth curves with appropriate changes.  

\subparagraph{Computing the simplification $\simp(\curve)$ of a piecewise smooth curve $\curve$ with respect to $\mu\in\RR$.}\label{algo:simplification}
Starting with the first smooth piece $\curve_1$ of $\curve$, consider its arc length $l(\curve_1)$. If $l(\curve_1)\ge \mu$, then we move on to the next piece $\curve_2$ and start over from there. 
Let $\curve_i:[0,1]\to\RR^d$ be the first piece with $l(\curve_i)<\mu$. Then $\curve_i$ is contained in the ball $B(\curve_i(0),\mu)$ of radius $\mu$ centered at $\curve_i(0)$. Following the curve $\curve$, let $\curve_j$ be the first piece of $\curve$ that leaves $B(\curve_i(0),\mu)$. We take the first point $p_j$ on $\curve_j$ lying on the boundary $B(\curve_i(0),\mu)$ and join $\curve_i(0)$ to $p_j$ with a straight line of length $\mu$, replacing the part of $\curve$ between $\curve_i(0)$ and $p_j$; and $\curve_j$ with the part of it after $p_j$. We then restart the procedure at $p_j$. If after any piece $\curve_k:[0,1]\to\RR^d$ with $l(\curve_k)<\mu$, $\curve$ does not intersect the boundary of $B(\curve_k(0),\mu)$, then we remove the part of $\curve$ after $\curve_k(0)$.

The simplification algorithm 
results in a piecewise smooth curve with each piece of having an arc length of at least $\mu$. 
Note that the algebraic operation of intersecting $\curve$ with a ball is the same operation required in the computation of the Fréchet distance between two curves.


\begin{lemma}\label{lem:simpdistance}
The simplification $\curve'=\simp(\curve,\mu)$ of a piecewise smooth curve defined by Algorithm~\ref{algo:simplification} above satisfies $\curve\subset \curve'\bigoplus B(0,\mu)$. In particular, $\delta_F(\curve',\curve)\le \mu$. 
\end{lemma}
\begin{figure}[t]
\centering 
\includegraphics{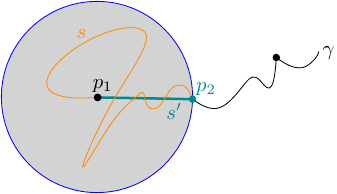}
    \caption{Illustration of Lemma~\ref{lem:simpdistance}, showing that the simplification $s'=\simp(s,\mu)$ is $\mu$-close to $s$.}\label{fig:simplification}
\end{figure}
Here, $A \bigoplus B$ denotes the Minkowski sum of $A$ and $B$. The situation is illustrated in Figure~\ref{fig:simplification}. 
\begin{proof}
Let $s'$ be the first new segment of $\curve'$ and $s$ the portion of $\curve$ that gets replaced by $s'$ in $\curve'$ and has the same endpoints $p_1$ and $p_2$ as $s'$. Now, $s$ is contained in $B(p_1,\mu)$, so $s\subset s' \bigoplus B(0,\mu)$ and the statement about the Minkowski sum follows, which in turn implies the statement about the Fréchet distance. Indeed, we can traverse $\curve'$ and $\curve$ until $p_1$, then traverse $\curve\cap B(p_1,\mu)$ while the parametrization of $\curve'$ stays on $p_1$. Then, $s'$ is traversed until $p_2$ is reached for both parametrizations, while maintaining a distance of at most $\mu$.
\end{proof}

\subsection{Preserving $c$-packedness}
We make precise the idea that since each piece of a simplified curve $\simp(\curve,\mu)$ is $\mu$-close to $\curve$, the length contained in any ball cannot increase too much through the simplification.
\begin{lemma}\label{lem:arclengthsimpbound}
Let $\curve'=\simp(\curve,\mu)$, with $\curve$ a piecewise smooth curve. Then for any $p\in \RR^d$ and $r\in \RR$, $l(\curve\cap B(p,r+\mu))\ge l(\curve'\cap B(p,r))$.
\end{lemma} 
The situation is illustrated in Figure~\ref{fig:arclengthbound}.  

\begin{figure}[t]
\centering 
\includegraphics{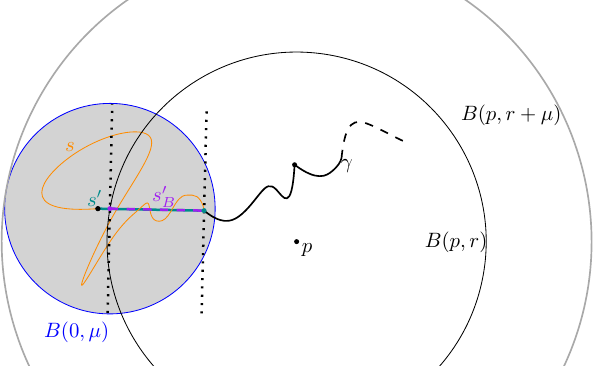}
    \caption{Illustration of the proof of Lemma~\ref{lem:arclengthsimpbound}.}\label{fig:arclengthbound}
\end{figure}

\begin{proof}
Let $s'$ be a segment of $\curve'$ and $s'_B:=s' \cap B(p,r)$. Consider the subcurve $s$ of $\curve$ that gets replaced by $s'$ in $\curve'$ and has the same endpoints as $s'$. By Lemma~\ref{lem:simpdistance}, $s\subset s' \bigoplus B(0,\mu)$. Moreover, $s'_B\bigoplus B(0,\mu)\subset B(p,r+\mu)$, so $s\cap \left(s'_B\bigoplus B(0,\mu)\right)\subset  B(p,r+\mu)$. Since $s'_B$ is a straight line, it is the shortest path connecting its two endpoints. Therefore, we can erect two planes orthogonal to $s'_B$ at the endpoints such that any path connecting these two planes must have length at least $l(s'_B)$. In particular, since  $s\cap s'_B\bigoplus B(0,\mu)$ contains such a path, $l(s\cap \left(s'_B\bigoplus B(0,\mu)\right))\ge l(s'_B)$. The statement of the lemma follows by summing over the intersections of all pieces of $\curve'$ with $B(p,r)$.
\end{proof}

\begin{lemma}\label{lem:simpcpacked}
Let $\curve$ be a $c$-packed planar smooth curve and $\mu>0$. Then the simplified curve $\curve'=\simp(\curve,\mu)$ is a $7c$-packed curve.
\end{lemma}

\begin{proof}
For the sake of contradiction, assume that $l(\curve'\cap B(p,r))> 7cr$. We consider first the case that $r\ge \mu$. Then Lemma~\ref{lem:arclengthsimpbound} implies that $l(\curve\cap B(p,r+\mu))\ge l(\curve\cap B(p,r+\mu))\ge l(\curve'\cap B(p,\mu))>7cr$. However, since $\curve$ is $c$-packed, we also have $l(\curve\cap B(p,r+\mu))\le (r+\mu)c   \le 2cr$, a contradiction. 

Let $r<\mu$. Let $k$ denote the number of curve pieces of $\curve'$ in $B(p,r)$ and consider $k=k_{n}+k_{o}$, where $k_o$ denotes the number of original pieces of $\curve'$ in $B(p,r)$ that are also present in $\curve\cap B(p,r)$ and $k_n$ denotes the number of new curve pieces that intersect $B(p,r)$. Since $\curve$ is $c$-packed, the curved pieces in $B(p,r)$ present in both $\curve$ and $\curve'$ cannot have a total length greater than $cr$. Therefore, if $l(\curve'\cap B(p,r))> 7cr$, then the new pieces have to add up in length to at least $6cr$. Since the new pieces are straight lines they can contribute at most $2r$ in length to the intersection $\curve'\cap B(p,r)$, so $k_n>6cr/2r=3c$.

Observe that $l(\curve'\cap B(p,2\mu))\ge l(\curve'\cap B(p,r+\mu))\ge k_n\mu$, since a new segment of $\curve'$ intersecting $B(p,r)$ also intersects $B(p,r+\mu)$ and each curve piece in $\curve'$ has length at least $\mu$. Therefore, by Lemma~\ref{lem:arclengthsimpbound}, $l(\curve \cap B(p,3\mu))\ge l(\curve' \cap B(p,2\mu))\ge l(\curve' \cap B(p,r+\mu))\ge k_n \mu> 3c\mu$, which contradicts the $c$-packedness of $\curve$.
\end{proof}
We note that the results and proofs in this section are valid for all simplifications of $c$-packed curves, as long as the simplified curve 
\begin{enumerate}
\item is contained in a tubular neighborhood of the original curve,
\item consists of pieces with arc length bounded from below by $\mu$.
\end{enumerate}

\section{The relative free space complexity of simplified smooth curves}\label{sec:freespacecomplexity}
The central idea of the approximation algorithm for the decision problem for piecewise smooth curves is to replace the two curves in question with their simplifications. To study the complexity of the reachable free space, we examine the following central quantity.
\begin{definition}\label{def:epsrelativefreespace}
The $\eps$-relative free space complexity of two curves $\curve_1$ and $\curve_2$ is   
\begin{align*}
N(\eps,\curve_1,\curve_2)=\max\limits_{\delta\ge 0}N_{\le \delta}(\simp(\curve_1,\eps\delta),\simp(\curve_2,\eps\delta)).
\end{align*}
\end{definition}

\begin{proposition}\label{prop:linearrelcomplexity}
Let $\curve_1$ and $\curve_2$ be two piecewise smooth curves that are $c$-packed and such that the arc length of each of the $n$ pieces is at least $\mu=\eps\delta$. Then the number of cells in $\freespacediagram_{\delta}(\curve_1,\curve_2)$ containing non-empty free space is in $O(cn/\eps)$.
\end{proposition}
\begin{proof}
Assign a counter to all pieces of the two curves, starting with the value $0$.
Each cell in $\freespace_{\le\delta}(\curve_1,\curve_2)$ has two associated smooth pieces $u_1$ and $u_2$ of $\curve_1$ and $\curve_2$, respectively. The free space in the cell associated to $u_1$ and $u_2$ is non-empty if and only if there are two points $p_1\in u_1 $ and $p_2\in u_2$ such that $\|p_1-p_2\|\le \delta$. For every cell in $\freespace_{\le \delta}$ with non-empty free space, we add one to the counter of the shorter of the two pieces. Figure~\ref{fig:counters} shows the pieces that contribute to the counter for $u_1$ in red. We bound the maximal possible value of the counter of each piece as follows.
\begin{figure}[t]
\centering 
\includegraphics{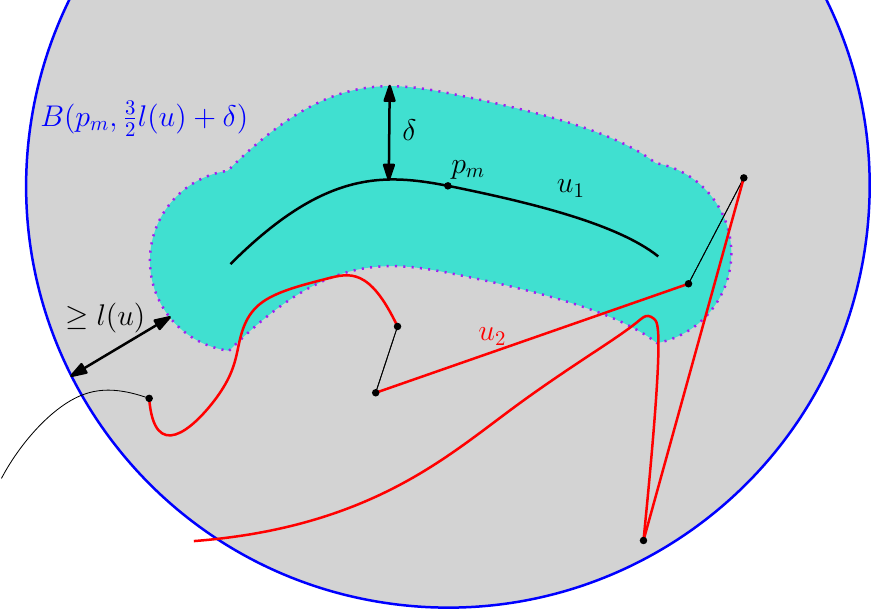}
    \caption{Illustration of the situation of Proposition~\ref{prop:linearrelcomplexity}.}\label{fig:counters}
\end{figure}

For a piece $u$ of $\curve_1$, consider the ball $B(p_m,(3/2)l(u)+\delta)$ of radius $(3/2)l(u)+\delta$ centered at its midpoint $p_m$. By construction, the distance of every point of $u$ to the boundary of $B(p_m,(3/2)l(u)+\delta)$ is at least $l(u)+\delta$. Therefore, every piece of $\curve_2$ that adds one to the counter of $u$ has an arc length of at least $l(u)$ in $B(p_m,(3/2)l(u)+\delta)$. Since $\curve_2$ is $c$-packed and $l(u)\ge \mu$, the number of times we add one to the counter of $u$ is bounded from above by
\begin{align*}
c'&=\frac{l(\curve_2)\cap B(p_m,(3/2)l(u)+\delta)}{l(u)}\le \frac{cr}{l(u)}=\frac{c(3/2l(u)+\delta)}{l(u)}= \frac 32 c + \frac{\delta c}{l(u)}\\
&\le \frac 32 c + \frac{\delta c}{\mu}=O\left(\frac{c}{\eps}\right). 
\end{align*} 
Thus, there are at most $O(cn/\eps)$ cells containing free space, concluding the proof of the proposition.
\end{proof}

\begin{theorem}\label{thm:simplelinearrelcomplexity}
Let $\curve_1$ and $\curve_2$ be two $c$-packed piecewise smooth curves and $0<\eps<1$. Then $N(\eps,\curve_1,\curve_2)=O(cn/\eps)$. 
\end{theorem}

\begin{proof}
The simplification of a curve with parameter $\mu=\delta\eps$ consists of pieces with arc length at least $\mu$. By Lemma~\ref{lem:simpcpacked}, the simplifications are moreover $7c$-packed. The claim thus follows from Proposition~\ref{prop:linearrelcomplexity}. 
\end{proof}

We can now apply the solution to the decision problem from Proposition~\ref{prop:decisionproblempath} to the decision problem for the simplified curves. Using this decision procedure in place of that for polygonal curves, the otherwise same algorithms (Lemmas~3.5 and~3.6) in~\cite{Driemel2012} can be used in our context to obtain an approximate solution to the decision problem for piecewise smooth $c$-packed curves that runs in $O(N(\eps,\curve_1,\curve_2))$ time. We note that $c$-packedness is used only in its role in guaranteeing that $N(\eps,\curve_1,\curve_2)=O(cn/\eps)$ by Theorem~\ref{thm:simplelinearrelcomplexity}.
\begin{corollary}\label{cor:decisionproblemsimpcurves}
Let $\curve_1$ and $\curve_2$ be two piecewise smooth algebraically bounded $c$-packed curves, $1\ge \eps >0$ and $\delta>0$. There is an algorithm that correctly outputs, in $O(cn/\eps)$ time, one of the following: 
\begin{enumerate}
    \item[(i)] a $(1+\eps)$-approximation to $\distf(\curve_1,\curve_2)$, 
    \item[(ii)] $\distf(\curve_1,\curve_2)<\delta$,
    \item[(iii)] $\distf(\curve_1,\curve_2)>\delta$.
\end{enumerate}
\end{corollary}

\bibliography{bibliography}

\bibliographystyle{plainurl}

\newpage
\appendix

\end{document}

%% file: abstract.tex
Since its introduction to computational geometry by Alt and Godau in 1992, the Fr\'echet distance has been a mainstay of algorithmic research on curve similarity computations. The focus of the research has been on comparing polygonal curves, with the notable exception of an algorithm for the decision problem for planar piecewise smooth curves due to Rote (2007). We present an algorithm for the decision problem for piecewise smooth curves that is both conceptually simpler and naturally extends to the first algorithm for the problem for piecewise smooth curves in $\RR^d$.

We assume that the algorithm is given two continuous curves, each consisting of a sequence of $m$, resp.\ $n$, smooth pieces, where each piece belongs to a sufficiently well-behaved class of curves, such as the set of algebraic curves of bounded degree. We introduce a decomposition of the free space diagram into a controlled number of pieces that can be used to solve the decision problem similarly to the polygonal case, in $O(mn)$ time, leading to a computation of the Fréchet distance that runs in $O(mn\log(mn))$ time.

Furthermore, we study approximation algorithms for piecewise smooth curves that are also $c$-packed for some fixed value $c$. 
We adapt the existing framework for $(1+\vareps)$-approximations and show that an approximate decision can be computed in $O(cn/\vareps)$ time for any $\vareps > 0$.


%% file: arxiv.bbl
\begin{thebibliography}{10}

\bibitem{ALT1995}
Helmut Alt and Michael Godau.
\newblock {Computing the Fr{\'{e}}chet distance between two polygonal curves}.
\newblock {\em International Journal of Computational Geometry and
  Applications}, 05(01n02):75--91, mar 1995.
\newblock \href {https://doi.org/10.1142/S0218195995000064}
  {\path{doi:10.1142/S0218195995000064}}.

\bibitem{Audin2014}
Mich{\`{e}}le Audin and Mihai Damian.
\newblock {\em {Morse Theory and Floer Homology}}.
\newblock Universitext. Springer London, London, 2014.
\newblock URL: \url{http://link.springer.com/10.1007/978-1-4471-5496-9}, \href
  {https://doi.org/10.1007/978-1-4471-5496-9}
  {\path{doi:10.1007/978-1-4471-5496-9}}.

\bibitem{Bates1993}
Sean~Michael Bates.
\newblock {Toward a Precise Smoothness Hypothesis in Sard's Theorem}.
\newblock {\em Proceedings of the American Mathematical Society}, 117(1):279,
  1993.
\newblock \href {https://doi.org/10.2307/2159728} {\path{doi:10.2307/2159728}}.

\bibitem{Bringmann2014}
Karl Bringmann.
\newblock {Why walking the dog takes time: Frechet distance has no strongly
  subquadratic algorithms unless SETH fails}.
\newblock In {\em Proceedings - Annual IEEE Symposium on Foundations of
  Computer Science, FOCS}, pages 661--670, 2014.
\newblock \href {http://arxiv.org/abs/1404.1448} {\path{arXiv:1404.1448}},
  \href {https://doi.org/10.1109/FOCS.2014.76}
  {\path{doi:10.1109/FOCS.2014.76}}.

\bibitem{chambers2010}
Erin~W. Chambers, {\'{E}}ric~Colin de~Verdi{\`{e}}re, Jeff Erickson, Sylvain
  Lazard, Francis Lazarus, and Shripad Thite.
\newblock Homotopic fr{\'{e}}chet distance between curves or, walking your dog
  in the woods in polynomial time.
\newblock {\em Comput. Geom.}, 43(3):295--311, 2010.
\newblock URL: \url{https://doi.org/10.1016/j.comgeo.2009.02.008}, \href
  {https://doi.org/10.1016/J.COMGEO.2009.02.008}
  {\path{doi:10.1016/J.COMGEO.2009.02.008}}.

\bibitem{Cole1987}
Richard Cole.
\newblock {Slowing down sorting networks to obtain faster sorting algorithms}.
\newblock {\em Journal of the ACM}, 34(1):200--208, jan 1987.
\newblock \href {https://doi.org/10.1145/7531.7537}
  {\path{doi:10.1145/7531.7537}}.

\bibitem{Driemel2012}
Anne Driemel, Sariel Har-Peled, and Carola Wenk.
\newblock {Approximating the Fr{\'{e}}chet distance for realistic curves in
  near linear time}.
\newblock {\em Discrete and Computational Geometry}, 48(1):94--127, 2012.
\newblock \href {http://arxiv.org/abs/1003.0460} {\path{arXiv:1003.0460}},
  \href {https://doi.org/10.1007/s00454-012-9402-z}
  {\path{doi:10.1007/s00454-012-9402-z}}.

\bibitem{Ferone2020}
Adele Ferone, Mikhail~V. Korobkov, and Alba Roviello.
\newblock {On some universal Morse–Sard type theorems}.
\newblock {\em Journal des Mathematiques Pures et Appliquees}, 139:1--34, 2020.
\newblock \href {https://doi.org/10.1016/j.matpur.2020.05.002}
  {\path{doi:10.1016/j.matpur.2020.05.002}}.

\bibitem{Guillemin2010}
Victor Guillemin and Alan Pollack.
\newblock {\em {Differential Topology}}, volume~4.
\newblock American Mathematical Society, Providence, Rhode Island, aug 2010.
\newblock URL: \url{http://www.ams.org/chel/370}, \href
  {https://doi.org/10.1090/chel/370} {\path{doi:10.1090/chel/370}}.

\bibitem{Hartshorne1977}
Robin Hartshorne.
\newblock {\em {Algebraic Geometry}}, volume~52 of {\em Graduate Texts in
  Mathematics}.
\newblock Springer New York, New York, NY, 1977.
\newblock \href {https://doi.org/10.1007/978-1-4757-3849-0}
  {\path{doi:10.1007/978-1-4757-3849-0}}.

\bibitem{Megiddo1983}
Nimrod Megiddo.
\newblock {Applying Parallel Computation Algorithms in the Design of Serial
  Algorithms}.
\newblock {\em Journal of the ACM}, 30(4):852--865, oct 1983.
\newblock \href {https://doi.org/10.1145/2157.322410}
  {\path{doi:10.1145/2157.322410}}.

\bibitem{Nirenberg2001}
Louis Nirenberg.
\newblock {\em {Topics in Nonlinear Functional Analysis}}, volume~6 of {\em
  Courant Lecture Notes}.
\newblock American Mathematical Society, Providence, Rhode Island, apr 2001.
\newblock \href {https://doi.org/10.1090/cln/006} {\path{doi:10.1090/cln/006}}.

\bibitem{Rote2007}
G{\"{u}}nter Rote.
\newblock {Computing the Fr{\'{e}}chet distance between piecewise smooth
  curves}.
\newblock {\em Computational Geometry: Theory and Applications}, 37(3 SPEC.
  ISS.):162--174, aug 2007.
\newblock \href {https://doi.org/10.1016/J.COMGEO.2005.01.004}
  {\path{doi:10.1016/J.COMGEO.2005.01.004}}.

\bibitem{Sard1942}
Arthur Sard.
\newblock {The measure of the critical values of differentiable maps}.
\newblock {\em Bulletin of the American Mathematical Society}, 48(12):883--890,
  1942.
\newblock \href {https://doi.org/10.1090/S0002-9904-1942-07811-6}
  {\path{doi:10.1090/S0002-9904-1942-07811-6}}.

\bibitem{oostrum2004}
Ren{\'{e}} van Oostrum and Remco~C. Veltkamp.
\newblock Parametric search made practical.
\newblock {\em Comput. Geom.}, 28(2-3):75--88, 2004.
\newblock URL: \url{https://doi.org/10.1016/j.comgeo.2004.03.006}, \href
  {https://doi.org/10.1016/J.COMGEO.2004.03.006}
  {\path{doi:10.1016/J.COMGEO.2004.03.006}}.

\end{thebibliography}
